\documentclass[11pt, a4paper, dvipdfmx]{article}

\usepackage{graphicx}
\usepackage{amsmath, amssymb, amsthm, amsfonts}
\usepackage{float, color}
\usepackage{ascmac, fancybox}
\usepackage{bm}

\def\N{{\mathbb N}}

\def\R{{\mathbb R}}  
\def\C{{\mathbb C}}

\def\H{{\mathcal{H}}}
\def\L{{\mathcal{L}}}
\def\S{{\mathcal{S}}}

\def\a{{\alpha}}
\def\b{{\beta}}
\def\ga{{\gamma}}

\def\e{{\varepsilon}}

\def\m{{\mu}}

\def\r{{\rho}}
\def\s{{\sigma}}
\def\t{{\tau}}
\def\th{{\theta}}
\def\x{{\xi}}

\def\Th{{\Theta}}
\def\o{{\omega}}
\def\O{{\Omega}}

\def\Tr{{\rm Tr\,}}

\def\D{\mathcal{D}}

\def\S{\mathcal{S}}
\def\H{\mathcal{H}}
\def\F{\mathcal{F}}

\def\R{\mathbb{R}}

\def\C{\mathbb{C}}
\def\braket#1#2{\left\langle #1|#2\right\rangle }

\def\L{\mathcal{L}}
\def\N{\mathbb{N}}
\def\convd#1{\overset{#1}{\rightsquigarrow}}

\def\Im{{\rm Im\,}}
\def\convq#1{\underset{#1}{\rightsquigarrow}}

\def\CCR#1{{\rm CCR}\left(#1\right)}
\def\supp{{\rm supp\,}}
\def\rank{{\rm rank\,}}

\newcommand{\ket}[1]{\left | #1 \right \rangle}
\newcommand{\bra}[1]{\left \langle #1 \right |}
\newcommand{\bracket}[2]{\left \langle #1 \left | #2 \right\rangle\right.}

\newtheorem{theorem}    {Theorem}
\newtheorem{lemma}      [theorem]{Lemma}
\newtheorem{corollary}  [theorem]{Corollary}
\newtheorem{remark}	[theorem]{Remark}
\newtheorem{define}     [theorem]{Definition}

\title{Noncommutative Lebesgue decomposition with application to quantum local asymptotic normality}
\author{Akio Fujiwara%
\thanks{fujiwara@math.sci.osaka-u.ac.jp}\\
{Department of Mathematics, Osaka University}\\ 
{Toyonaka, Osaka 560-0043, Japan}\\ \\
and \\ \\
Koichi Yamagata%
\thanks{yamagata.28m@chuo-u.ac.jp}\\
{Department of Information and System Engineering, Chuo University} \\
{Bunkyo-ku, Tokyo 112-8551, Japan}
}%
\date{\today}

\begin{document}
\maketitle

\begin{abstract}
We develop a theory of local asymptotic normality in the quantum domain based on a noncommutative extension of the Lebesgue decomposition. 
This formulation gives a substantial generalization of the previous paper [Yamagata, Fujiwara, and Gill (2013).  \textit{Ann. Statist.},  \textbf{41}, 2197-2217.], extending the scope of the quantum local asymptotic normality to a wider class of quantum statistical models that comprise density operators of mixed ranks.\end{abstract}

\section{Introduction}

In \cite{YFG}, we formulated a theory of quantum local asymptotic normality (q-LAN) for quantum statistical models that comprise mutually absolutely continuous density operators on a finite dimensional Hilbert space 
$\H$.  
Here, density operators $\r$ and $\s$ are said to be {\em mutually absolutely continuous}, $\r\sim\s$ in symbols, 
if there exists a Hermitian operator $\L$ that satisfies
\[
  \s=e^{\frac{1}{2}\L}\r e^{\frac{1}{2}\L}.
\]
The operator $\L$ satisfying this relation is called (a version of) the {\em quantum log-likelihood ratio}  \cite{YFG}. 
We might as well call the operator $\L$ the {\em symmetric log-likelihood ratio} by analogy with the term symmetric logarithmic derivative \cite{{Helstrom:1976}, {Holevo:1982}}.
When the reference states $\r$ and $\s$ need to be specified, $\L$ is denoted by $\L(\s |\r)$, so that 
\[
  \s=e^{\frac{1}{2}\L(\s |\r)} \r e^{\frac{1}{2}\L(\s |\r)}.
\]
We use the convention that $\L(\r |\r)=0$.

For example, when both $\r$ and $\s$ are strictly positive, the quantum log-likelihood ratio is uniquely given by
\[
 \L(\s |\r)=2\log\left(\s \# \r^{-1} \right).
\]
Here, the operator geometric mean $A\# B$ \cite{{Bhatia},{KuboAndo}} for strictly positive operators $A$ and $B$ is defined as the positive operator $X$ satisfying the equation $B=XA^{-1}X$, 
and is explicitly given by
$\displaystyle A\# B=\sqrt{A}\sqrt{\sqrt{A^{-1}} B \sqrt{A^{-1}}\,} \sqrt{A}$.

The theory of q-LAN developed in \cite{YFG} was based essentially on the analysis of the quantum log-likelihood ratio; 
thus the assumption of mutual absolute continuity for quantum statistical models to be investigated appears indispensable.  
Nevertheless, the original definition of classical LAN did not require mutual absolute continuity for the model \cite{Vaart}: 
a sequence $\left\{P_{\theta}^{(n)} \left|\;\theta\in\Theta\subset\R^{d}\right.\right\}$ 
of $d$-dimensional statistical models, 
each comprising probability measures on a measurable space $(\Omega^{(n)}, \F^{(n)})$, 
is said to be {\em locally asymptotically normal} at $\theta_{0}\in\Theta$ 
if there exist a sequence $\Delta^{(n)}=(\Delta_1^{(n)},\,\dots,\,\Delta_d^{(n)})$ of $d$-dimensional random vectors and a $d\times d$ positive definite matrix $J$ such that $\Delta^{(n)}\convd {0} N(0,J)$ and
\[
\log\frac{dP_{\theta_{0}+h/\sqrt{n}}^{(n)}}{dP_{\theta_{0}}^{(n)}}
=h^{i}\Delta_{i}^{(n)}-\frac{1}{2}h^{i}h^{j}J_{ij}+o_{P_{\theta_0}^{(n)}}(1),
\qquad 
(h\in\R^{d}).
\]
Here the arrow $\convd {h}$ stands for the convergence in distribution under $P_{\theta_{0}+h/\sqrt{n}}^{(n)}$, the remainder term $o_{P_{\theta_0}^{(n)}}(1)$ converges in probability to zero under $P_{\theta_0}^{(n)}$, and Einstein's summation convention is used.

The key idea behind this classical formulation is the use of the Radon-Nikodym density, 
or more fundamentally, the use of the Lebesgue decomposition of $P_{\theta_{0}+h/\sqrt{n}}^{(n)}$ with respect to $P_{\theta_{0}}^{(n)}$.
In order to extend such a flexible formulation to the quantum domain, we must invoke 
 a proper quantum analogue of the Lebesgue decomposition.
However, no such analogue that is applicable to the theory of q-LAN
is known to date. 

The objective of the present paper is twofold: 
we first devise a theory of the Lebesgue decomposition in the quantum domain 
that is consistent with the framework of \cite{YFG}, 
and then generalize the theory of q-LAN in order to get rid of the assumption of mutual absolute continuity for the model.

The paper is organized as follows. 
In Section \ref{sec:acANDsing}, we extend the absolute continuity and singularity to the quantum domain 
in such a way that they are fully consistent with the notion of quantum mutual absolute continuity introduced in \cite{YFG}. 
By exploiting these notions, we formulate a noncommutative analogue of the Lebesgue decomposition in Section \ref{sec:Lebesgue}. 
In Section \ref{sec:qLAN},  we develop a theory of q-LAN that enables us to treat quantum statistical models comprising density operators of mixed ranks. 
In Section \ref{sec:example}, we give a simple illustrative example to demonstrate the flexibility of our framework. 
Section \ref{sec:conclusion} is devoted to concluding remarks. 
Throughout the paper, we assume some familiarity with terms and notations introduced in \cite{YFG}, and therefore, we give a brief overview of them in Appendix for the reader's convenience.

\section{Absolute continuity and singularity}\label{sec:acANDsing}

Given positive operators $\r$ and $\s$ on a finite dimensional Hilbert space $\H$ with $\r\neq 0$, 
let $\sigma\!\!\downharpoonleft_{\supp\rho}$ denote the {\em excision}
of $\s$ relative to $\r$ by the operator on the subspace $\supp\rho:=(\ker\rho)^\perp$ of $\H$ defined by
\[ \sigma\!\!\downharpoonleft_{\supp\rho}:=\iota_\rho^*\, \sigma\, \iota_\rho, \]
where $\iota_\rho: \supp\rho\hookrightarrow \H$ is the inclusion map. 
More specifically, let
\begin{equation}\label{eqn:blockRS}
\rho=\begin{pmatrix} \rho_0 & 0\\ 0 & 0 \end{pmatrix},
\qquad
\sigma=\begin{pmatrix} \sigma_0 & \alpha\\ \alpha^* & \beta
\end{pmatrix}
\end{equation}
be a simultaneous block matrix representations of $\r$ and $\s$, where $\r_0>0$. 
Then the excision $\sigma\!\!\downharpoonleft_{\supp\rho}$ is nothing but the operator represented by the $(1,1)$-block $\s_0$ of $\s$. 
The notion of the excision was usefully exploited in \cite{YFG}. 
In particular, it was shown that $\r$ and $\s$ are mutually absolutely continuous if and only if 
\[
 \sigma\!\!\downharpoonleft_{\supp\rho}>0 \quad\mbox{and}\quad \rank\rho=\rank\sigma,
\]
or equivalently, if and only if
\begin{equation}\label{eqn:mutuallyAC}
 \sigma\!\!\downharpoonleft_{\supp\rho}>0 \quad\mbox{and}\quad \rho\!\!\downharpoonleft_{\supp\sigma}>0.
\end{equation}

Now we introduce noncommutative analogues of the absolute continuity and singularity that played essential roles in the classical measure theory. 
Given positive operators $\r$ and $\s$, we say $\r$ is {\em singular} with respect to $\s$, denoted by $\r\perp\s$, if
\[
 \sigma\!\!\downharpoonleft_{\supp\rho}=0.
\]
The following lemma implies that the relation $\perp$ is symmetric; 
this fact allows us to say that $\r$ and $\s$ are mutually singular, as in the classical case. 

\begin{lemma}\label{lem:1}
For nonzero positive operators $\r$ and $\s$, the following are equivalent. 
\begin{itemize}
\item[{\rm (a)}]  $\r\perp\s$. 
\item[{\rm (b)}]  $\supp\rho \perp \supp\sigma$.
\item[{\rm (c)}]  $\Tr \r \s=0$.
\end{itemize}
\end{lemma}

\begin{proof}
Let us represent $\r$ and $\s$ in the form \eqref{eqn:blockRS}. 
Then, (a) is equivalent to $\s_0=0$. 
In this case, the positivity of $\s$ entails that the off-diagonal blocks $\a$ and $\a^*$ of $\s$ must vanish, and $\s$ takes the form
\[ 
\sigma=\begin{pmatrix}
0 & 0\\
0 & \beta
\end{pmatrix}.
\]
This implies (b). 
Next, (b) $\Rightarrow$ (c) is obvious.  
Finally, assume (c). 
With the representation \eqref{eqn:blockRS}, this is equivalent to $\Tr \r_0 \s_0=0$.
Since $\r_0>0$, we have $\s_0=0$, proving (a).
\end{proof}

We next introduce the notion of absolute continuity. 
Given positive operators $\r$ and $\s$, we say $\r$ is {\em absolutely continuous} with respect to $\s$, denoted by $\r\ll\s$,  if
\[
 \sigma\!\!\downharpoonleft_{\supp\rho}>0.
\]

Some remarks are in order. 
First, the above definition of absolute continuity is consistent with the definition of mutual absolute continuity:  
in fact, as demonstrated in (\ref{eqn:mutuallyAC}), $\r$ and $\s$ are mutually absolutely continuous if and only if both $\r\ll\s$ and $\s\ll\r$ hold. 
Second, $\r\ll\s$ is a much weaker condition than $\supp\r\subset\supp\s$; 
this makes a striking contrast to the classical measure theory. 
For example, pure states $\r=\ket{\psi}\bra{\psi}$ and $\s=\ket{\xi}\bra{\xi}$ are mutually absolutely continuous if and only if $\braket{\xi}{\psi}\neq 0$, (see \cite[Example 2.3]{YFG}).

The next lemma plays an essential role in the present paper. 

\begin{lemma}\label{lem:2}
For nonzero positive operators $\r$ and $\s$, the following are equivalent. 
\begin{itemize}
\item[{\rm (a)}]  $\r\ll\s$. 
\item[{\rm (b)}]  $\exists R> 0$ such that $\s\ge R\r R$.
\item[{\rm (c)}]  $\exists R> 0$ such that $\r\le R\s R$.
\item[{\rm (d)}]  $\exists R\ge 0$ such that $\r=R\s R$.
\end{itemize}
\end{lemma}

\begin{proof}
We first prove (a) $\Rightarrow$ (b). 
Let
\[ 
\rho=\begin{pmatrix} \rho_0 & 0\\ 0 & 0 \end{pmatrix},
\qquad
\sigma=\begin{pmatrix} \sigma_0 & \alpha\\ \alpha^* & \beta
\end{pmatrix}
\]
where $\r_0>0$. 
Since $\s_0=\s\!\!\downharpoonleft_{\supp\r}>0$, 
the matrix $\s$ is further decomposed as
\[
\sigma=
E^*
\begin{pmatrix}
\sigma_0 & 0\\
0 & \beta-\alpha^* \sigma_0^{-1} \alpha
\end{pmatrix}
E,
\qquad
E:=
\begin{pmatrix}
I & \sigma_0^{-1} \alpha\\
0& I
\end{pmatrix}.
\]
Note that, since $\s\ge 0$ and $E$ is full-rank, we have
\begin{equation}\label{eqn:sigma22}
\beta-\alpha^* \sigma_0^{-1} \alpha\ge 0.
\end{equation} 
Now we set
\[
R:=E^* \begin{pmatrix} X & 0\\ 0 & \ga \end{pmatrix} E,
\]
where $X:=\s_0 \# \r_0^{-1}$, 
and $\ga$ is an arbitrary strictly positive operator. 
Then 
\begin{eqnarray*}
R\r R
&=& 
E^* \begin{pmatrix} X & 0\\ 0 & \ga \end{pmatrix} E
\begin{pmatrix} \rho_0 & 0\\ 0 & 0 \end{pmatrix}
E^* \begin{pmatrix} X & 0\\ 0 & \ga \end{pmatrix} E \\
&=& 
E^* \begin{pmatrix} X & 0\\ 0 & \ga \end{pmatrix}
\begin{pmatrix} \rho_0 & 0\\ 0 & 0 \end{pmatrix}
\begin{pmatrix} X & 0\\ 0 & \ga \end{pmatrix} E \\
&=& 
E^* \begin{pmatrix} X \rho_0 X & 0\\ 0 & 0 \end{pmatrix} E \\
&=& 
E^* \begin{pmatrix} \s_0 & 0\\ 0 & 0 \end{pmatrix} E \\
&\le& 
E^* \begin{pmatrix} \s_0 & 0\\ 0 & \beta-\alpha^* \sigma_0^{-1} \alpha \end{pmatrix} E
=
\s.
\end{eqnarray*}
Here, the inequality is due to (\ref{eqn:sigma22}). 
Since $R>0$, we have (b).

We next prove (b) $\Rightarrow$ (a). 
Due to assumption, there is a positive operator $\t\ge 0$ such that 
\[ \s=R\r R+\t. \]
Let
\[ 
\rho=\begin{pmatrix} \rho_0 & 0\\ 0 & 0 \end{pmatrix},\qquad
R=\begin{pmatrix} R_0 & R_1\\ R_1^* & R_2 \end{pmatrix},\qquad
\t=\begin{pmatrix} \t_0 & \t_1\\ \t_1^* & \t_2 \end{pmatrix},
\]
where $\r_0>0$. 
Then
\[
 \s=\begin{pmatrix} R_0\r_0 R_0+\t_0 & R_0\r_0 R_1 +\t_1
 	\\ R_1^*\r_0 R_0+\t_1^* & R_1^*\r_0 R_1+\t_2 \end{pmatrix}
\]
and
\[
 \s\!\!\downharpoonleft_{\supp\r}=R_0\r_0 R_0+\t_0. 
\]
Since $R_0>0$ and $\t_0\ge 0$, we have $\s\!\!\downharpoonleft_{\supp\r}>0$.

Now that the equivalence (b) $\Leftrightarrow$ (c) is obvious, 
we proceed to the proof of (a) $\Rightarrow$ (d). Let
\[ 
\r=\begin{pmatrix} \r_0 & 0\\ 0 & 0 \end{pmatrix},
\qquad
\s=\begin{pmatrix} \s_0 & \a \\ \a^* & \b
\end{pmatrix}, 
\]
where $\r_0>0$. Since $\s_0=\s\!\!\downharpoonleft_{\supp\r}>0$, 
\[
R:=\begin{pmatrix} \r_0 \# \s_0^{-1} & 0\\ 0 & 0 \end{pmatrix}
\]
is a well-defined positive operator satisfying
\[
 \r=R \s R. 
\]
This proves (d). 

Finally, we prove (d) $\Rightarrow$ (a). 
Let the positive operator $R$ in $\r=R \s R$ be represented as
\[ 
R=\begin{pmatrix} R_0 & 0\\ 0 & 0 \end{pmatrix},
\]
where $R_0>0$, and accordingly, let us represent $\r$ and $\s$ as
\[ 
\r=\begin{pmatrix} \r_0 & \r_1\\ \r_1^* & \r_2 \end{pmatrix},
\qquad
\s=\begin{pmatrix} \s_0 & \s_1 \\ \s_1^* & \s_2
\end{pmatrix}.
\]
The relation $\r=R \s R$ is then reduced to 
\[ 
\begin{pmatrix} \r_0 & \r_1\\ \r_1^* & \r_2 \end{pmatrix}
=
\begin{pmatrix} R_0 \s_0 R_0 & 0 \\ 0 & 0
\end{pmatrix}.
\]
This implies that $\supp\r=\supp\r_0$ and $\r_0\sim\s_0$. 
Consequently, 
\[
 \s\!\!\downharpoonleft_{\supp\r}
 =\s\!\!\downharpoonleft_{\supp\r_0}
 =\s_0\!\!\downharpoonleft_{\supp\r_0}
 \,>0.
\]
In the last inequality, we used the fact that $\r_0\sim\s_0$ implies $\r_0\ll\s_0$. 
\end{proof}

\section{Lebesgue decomposition}\label{sec:Lebesgue}

In this section, we extend the Lebegue decomposition to the quantum domain.

\subsection{Case 1:  when $\s \gg \r$}

To elucidate our motivation, 
let us first treat the case when $\s\gg \r$. 
In Lemma \ref{lem:2}, we found the following characterization: 
\[
 \s\gg\r \;\Longleftrightarrow\; \exists R> 0
 \mbox{ such that $\s\ge R\r R$}.
\]
Note that such an operator $R$ is not unique. 
For example, suppose that $\s\ge R_1\r R_1$ holds for some $R_1>0$.
Then for any $t\in (0,1]$, the operator $R_t:=t R_1$ is strictly positive and satisfies $\s\ge R_t \r R_t$. 
It is then natural to seek, if any, the ``maximal'' operator of the form $R\r R$ that is packed into $\s$.  
Put differently, letting $\t:=\s-R\r R$, we want to find the ``minimal'' positive operator $\t$ that satisfies 
\begin{equation}\label{eqn:decomp0}
 \s=R\r R+\t,
\end{equation}
where $R>0$.
This question naturally leads us to a noncommutative analogue of the Lebesgue decomposition, 
in that 
a positive operator $\t$ satisfying (\ref{eqn:decomp0}) is regarded as minimal if $\t \perp \r$. 

In the proof of Lemma \ref{lem:2}, we found the following decomposition: 
\begin{eqnarray*}
\s
&=&E^* \begin{pmatrix} \s_0 & 0\\ 0 & \beta-\alpha^* \sigma_0^{-1} \alpha \end{pmatrix} E \\
&=&E^* \begin{pmatrix} \s_0 & 0\\ 0 & 0  \end{pmatrix} E
+E^* \begin{pmatrix} 0 & 0\\ 0 & \beta-\alpha^* \sigma_0^{-1} \alpha \end{pmatrix} E \\
&=&R\r R+\begin{pmatrix} 0 & 0\\ 0 & \beta-\alpha^* \sigma_0^{-1} \alpha \end{pmatrix}
\end{eqnarray*}
where
\[
\rho=\begin{pmatrix} \r_0 & 0\\ 0 & 0 \end{pmatrix},
\quad
\sigma=\begin{pmatrix} \s_0 & \a\\ \a^* & \beta \end{pmatrix},
\quad
E:=\begin{pmatrix} I & \s_0^{-1} \a\\ 0& I \end{pmatrix},
\quad 
R=E^* \begin{pmatrix} \s_0 \# \r_0^{-1} & 0\\ 0 & I \end{pmatrix} E
\]
with $\r_0>0$ and $\s_0>0$. 
Since
\[
 \begin{pmatrix} \r_0 & 0\\ 0 & 0 \end{pmatrix}
 \perp
 \begin{pmatrix} 0 & 0\\ 0 & \beta-\alpha^* \sigma_0^{-1} \alpha \end{pmatrix},
\]
we have the following decomposition:
\begin{equation}\label{eqn:LebDec1}
 \s=\s^a+\s^\perp,
\end{equation}
where
\begin{equation}\label{eqn:LebDec1AC}
\s^a:=R\r R=\begin{pmatrix} \s_0 & \a \\ \a^* & \alpha^* \sigma_0^{-1} \alpha \end{pmatrix}
\end{equation}
is the (mutually) {\em absolutely continuous part }of $\s$ with respect to  $\r$, 
and 
\begin{equation}\label{eqn:LebDec1singular}
\s^\perp:=\begin{pmatrix} 0 & 0\\ 0 & \beta-\alpha^* \sigma_0^{-1} \alpha \end{pmatrix}
\end{equation}
is the {\em singular part} of $\s$ with respect to  $\r$. 

We shall call the decomposition \eqref{eqn:LebDec1} a {\em quantum Lebesgue decomposition} for the following reasons. 
First, 
although (\ref{eqn:LebDec1}) was defined by using a simultaneous block matrix representation of $\r$ and $\s$, which has an arbitrariness of unitary transformations of the form $U_1\oplus U_2$, 
the matrices (\ref{eqn:LebDec1AC}) and (\ref{eqn:LebDec1singular}) are covariant under those unitary transformations, and hence 
the operators $\s^a$ and $\s^\perp$ are well-defined regardless of the arbitrariness of the block matrix representation. 
Second, 
the decomposition (\ref{eqn:LebDec1}) is unique, as the following lemma asserts. 

\begin{lemma}\label{lem:uniqueness1}
Suppose $\s\gg\r$. Then the decomposition
\begin{equation}\label{eqn:qLebesgue1}
 \s=\s^a+\s^\perp\qquad (\s^a \ll \r,\;\s^\perp \perp \r)
\end{equation}
is uniquely given by (\ref{eqn:LebDec1AC}) and (\ref{eqn:LebDec1singular}).
\end{lemma}

\begin{proof}
We show that the decomposition
\begin{equation}\label{eqn:uniqueness1}
 \s=R\r R+\t\qquad (R \ge 0,\,\t\ge 0,\; \t \perp \r)
\end{equation}
is unique. 
Let
\[
\rho=\begin{pmatrix} \r_0 & 0\\ 0 & 0 \end{pmatrix},
\quad
\sigma=\begin{pmatrix} \s_0 & \a\\ \a^* & \beta \end{pmatrix}
\]
with $\r_0>0$. 
Due to assumption $\r\ll\s$, we have $\s_0>0$. 
Let
\[
E:=\begin{pmatrix} I & \s_0^{-1} \a\\ 0& I \end{pmatrix}.
\]
Since $E$ is invertible, the operator $R$ appeared in (\ref{eqn:uniqueness1}) is represented as
\[
R=E^* \begin{pmatrix} R_0 & R_1 \\ R_1^* & R_2 \end{pmatrix} E.
\]
With this representation
\begin{eqnarray*}
R\r R
&=& 
E^* \begin{pmatrix} R_0 & R_1 \\ R_1^* & R_2  \end{pmatrix} E
\begin{pmatrix} \rho_0 & 0\\ 0 & 0 \end{pmatrix}
E^* \begin{pmatrix} R_0 & R_1 \\ R_1^* & R_2 \end{pmatrix} E \\
&=& 
E^* 
\begin{pmatrix} R_0 \r_0 R_0 & R_0\r_0 R_1 \\ R_1^*\r_0 R_0 & R_1^*\r_0 R_1  \end{pmatrix} 
E\\
&\le& 
\s
=E^* \begin{pmatrix} \s_0 & 0\\ 0 & \beta-\alpha^* \sigma_0^{-1} \alpha \end{pmatrix} E.
\end{eqnarray*}
Here, the inequality is due to (\ref{eqn:uniqueness1}). 
Let us denote the singular part $\t$ as
\[
 \t=\begin{pmatrix} 0 & 0 \\ 0 & \t_0 \end{pmatrix}
 =E^* \begin{pmatrix} 0 & 0 \\ 0 & \t_0 \end{pmatrix} E.
\]
Then the decomposition (\ref{eqn:uniqueness1}) is equivalent to 
\begin{equation}\label{eqn:uniqueness1-1}
 \begin{pmatrix} \s_0 & 0\\ 0 & \beta-\alpha^* \sigma_0^{-1} \alpha \end{pmatrix} 
 = \begin{pmatrix} R_0 \r_0 R_0 & R_0\r_0 R_1 \\ R_1^*\r_0 R_0 & R_1^*\r_0 R_1  \end{pmatrix} 
 + \begin{pmatrix} 0 & 0 \\ 0 & \t_0 \end{pmatrix}.
\end{equation}
Comparison of the $(1,1)$th blocks of both sides yields 
$R_0=\s_0 \# \r_0^{-1}$. 
Since this $R_0$ is strictly positive, comparison of other blocks of \eqref{eqn:uniqueness1-1} further yields
\[
 R_1=0\quad\mbox{and}\quad \t_0=\beta-\alpha^* \sigma_0^{-1} \alpha.
\]
Consequently, the singular part $\t$ is uniquely determined by (\ref{eqn:LebDec1singular}). 
\end{proof}

An immediate consequence of Lemma \ref{lem:uniqueness1} is the following

\begin{corollary}\label{cor:uniqueness1}
When $\s\gg\r$, the absolutely continuous part $\s^a$ of the quantum Lebesgue decomposition \eqref{eqn:qLebesgue1}
is in fact mutually absolutely continuous to $\r$, i.e., $\s^a\sim\r$.
\end{corollary}

Note that the operator $R_2$ appeared in the proof of Lemma \ref{lem:uniqueness1} is arbitrary as long as it is positive. 
Because of this arbitrariness, we can take the operator $R$ in (\ref{eqn:uniqueness1}) to be strictly positive. 
This gives an alternative view of Corollary \ref{cor:uniqueness1}.

\subsection{Case 2:  generic case}

Let us extend the quantum Lebesgue decomposition \eqref{eqn:qLebesgue1} to a generic case when 
$\r$ is not necessarily absolutely continuous with respect to $\s$. 
When $\r$ and $\s$ are mutually singular, however, we just let $\s^a=0$ and $\s^\perp=\s$. 
In the rest of this section, therefore, we assume that $\r$ and $\s$ are not mutually singular. 

Given positive operators $\r$ and $\s$ that satisfy $\r\not\perp\s$, 
let $\H=\H_1\oplus\H_2\oplus\H_3$ be the orthogonal direct sum decomposition defined by
\[ 
 \H_1:=\ker\left(\sigma\!\!\downharpoonleft_{\supp\rho}\right),\qquad 
 \H_2:=\supp\left(\sigma\!\!\downharpoonleft_{\supp\rho}\right),\qquad 
 \H_3:=\ker\r. 
\]
Then $\r$ and $\s$ are represented in the form of block matrices as follows:
\begin{equation}\label{eqn:blockMatrix}
\rho=\begin{pmatrix} \r_2 & \r_1 & 0\\ \r_1^* & \r_0 & 0 \\ 0 & 0 & 0 \end{pmatrix}, 
\qquad
\sigma=\begin{pmatrix} 0 & 0 & 0 \\ 0 & \s_0 & \a \\ 0 & \a^* & \beta \end{pmatrix},
\end{equation}
where 
\[
\begin{pmatrix} \r_2 & \r_1 \\ \r_1^* & \r_0 \end{pmatrix}>0,\qquad \s_0>0.
\]
Note that when 
$\s \gg \r$ (Case 1), the subspace $\H_1$ becomes zero; in this case, 
the first rows and columns in \eqref{eqn:blockMatrix} should be ignored.
Likewise, when $\r>0$, the subspace $\H_3$ becomes zero; in this case, 
the third rows and columns in \eqref{eqn:blockMatrix} should be ignored. 

\begin{figure}[t] 
	\begin{centering}
	\includegraphics[scale=0.3]{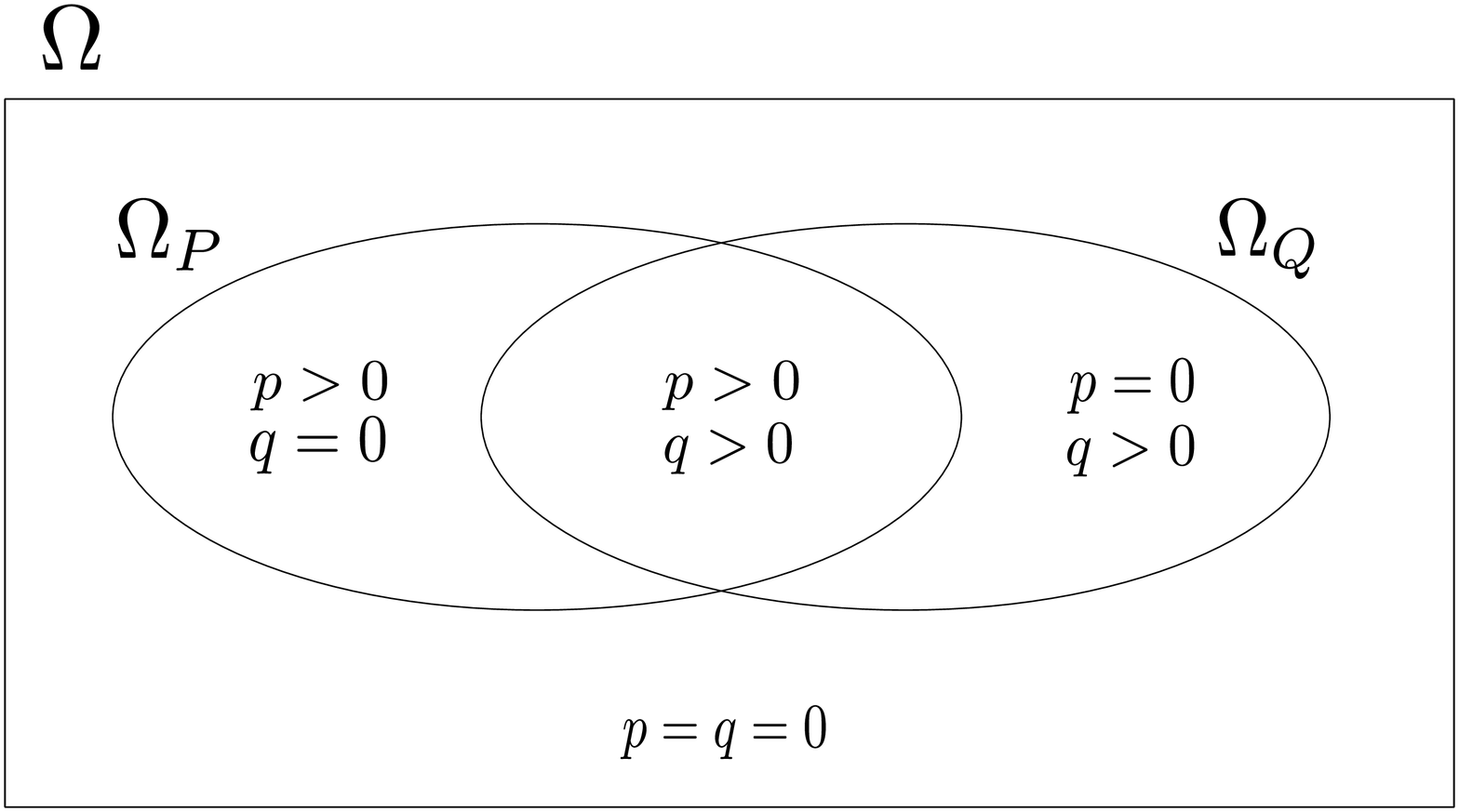}
	\par
	\end{centering}
	\caption{
	Schematic diagram of support sets of measures $P$ and $Q$ on a classical measure space 
	$(\O,\F,\m)$ having densities $p$ and $q$, respectively. 
	Here $\O_P:=\{\o\in\O\,|\, p(\o)>0\}$ and $\O_Q:=\{\o\in\O\,|\, q(\o)>0\}$. 
	The induced measures $Q^a(A):=Q(A\cap \{p>0\})$ and $Q^\perp(A):=Q(A\cap \{p=0\})$ give the 
	Lebesgue decomposition $Q=Q^a+Q^\perp$ with respect to $P$, 
	in which $Q^a\ll P$ and $Q^\perp\perp P$, (cf. \cite[Chapter 6]{Vaart}).
	\label{fig:support}}
\end{figure}

There is an obvious resemblance between the block matrix structure in \eqref{eqn:blockMatrix} 
and the diagram depicted in Fig.~\ref{fig:support} that displays the support sets of two measures $P$ and $Q$ on a classical measure space 
$(\O,\F,\m)$ having densities $p$ and $q$, respectively. 
However,  it should be warned that
\[
 \H_1':=\supp\r\cap \ker\s,\qquad \H_2':=\supp\r\cap\supp\s
\]
are different from $\H_1$ and $\H_2$, respectively. 
This is most easily seen by considering the case when both $\r$ and $\s$ are pure states: 
for pure states $\r=\ket{\psi}\bra{\psi}$ and $\s=\ket{\xi}\bra{\xi}$, we see that $\H_2\neq \{0\}$ if and only if $\braket{\xi}{\psi}\neq 0$, (cf. \cite[Example 2.3]{YFG}), while $\H_2'\neq\{0\}$ if and only if $\r=\s$.

Let us rewrite $\s$ in the form
\[
\sigma
=E^*
\begin{pmatrix}
0 & 0 & 0\\
0 & \s_0 & 0 \\
0 & 0 & \beta-\alpha^* \sigma_0^{-1} \alpha
\end{pmatrix}
E,
\]
where
\[
E:=
\begin{pmatrix}
I & 0 & 0\\
0 & I & \s_0^{-1} \a \\
0 & 0 & I
\end{pmatrix}.
\]
Since $E$ is invertible and $\s\ge 0$, we see that
\[
\beta-\alpha^* \sigma_0^{-1} \alpha\ge 0.
\]
Now let 
\[
\s^a
:=E^*
\begin{pmatrix} 0 & 0 & 0\\ 0 & \s_0 & 0 \\ 0 & 0 & 0\end{pmatrix}
E
=\begin{pmatrix} 0 & 0 & 0  \\ 0 & \s_0 & \a \\ 0 & \a^* & \a^* \s_0^{-1} \a \end{pmatrix}
\]
and let
\[
\s^\perp:=
E^*
\begin{pmatrix} 0 & 0 & 0 \\ 0 & 0 & 0\\ 0 & 0 & \b-\a^* \s_0^{-1} \a \end{pmatrix}
E
=\begin{pmatrix} 0 & 0 & 0 \\ 0 & 0 & 0\\ 0 & 0 & \b-\a^* \s_0^{-1} \a \end{pmatrix}.
\]
Then it is shown that $\s^a\ll\r$ and $\s^\perp \perp \r$.
In fact, the latter is obvious from Lemma \ref{lem:1}. 
To prove the former, let
\[
R:=E^* \begin{pmatrix} 0 & 0 & 0 \\ 0 & \s_0 \# \r_0^{-1} & 0 \\ 0 & 0 & 0 \end{pmatrix} E.
\]
Then $R$ is a positive operator satisfying
\begin{eqnarray*}
R\r R
&= &
E^*
\begin{pmatrix} 0 & 0 & 0 \\ 0 & \s_0 \# \r_0^{-1} & 0 \\ 0 & 0 & 0 \end{pmatrix}
\begin{pmatrix} \r_2 & \r_1 & 0\\ \r_1^* & \r_0 & 0 \\ 0 & 0 & 0 \end{pmatrix}
\begin{pmatrix} 0 & 0 & 0 \\ 0 & \s_0 \# \r_0^{-1} & 0 \\ 0 & 0 & 0 \end{pmatrix}
E \\
&= &
E^*
\begin{pmatrix} 0 & 0 & 0\\ 0 & \s_0 & 0 \\ 0 & 0 & 0\end{pmatrix}
E
=
\s^a.
\end{eqnarray*}
It then follows from Lemma \ref{lem:2} that $\s^a\ll\r$.

In summary, given $\r$ and $\s$ that satisfy $\s\not\perp\r$, let 
\begin{equation}\label{eqn:simultaneousBlock}
\rho=\begin{pmatrix} \r_2 & \r_1 & 0\\ \r_1^* & \r_0 & 0 \\ 0 & 0 & 0 \end{pmatrix}, 
\qquad
\sigma=\begin{pmatrix} 0 & 0 & 0 \\ 0 & \s_0 & \a \\ 0 & \a^* & \beta \end{pmatrix}
\end{equation}
be their simultaneous block matrix representations, where
\[
\begin{pmatrix} \r_2 & \r_1 \\ \r_1^* & \r_0 \end{pmatrix}>0,\qquad \s_0>0.
\]
Then
\begin{equation}\label{eqn:LebDec2}
\s^a
=\begin{pmatrix} 0 & 0 & 0  \\ 0 & \s_0 & \a \\ 0 & \a^* & \a^* \s_0^{-1} \a \end{pmatrix},
\qquad
\s^\perp
=\begin{pmatrix} 0 & 0 & 0 \\ 0 & 0 & 0\\ 0 & 0 & \b-\a^* \s_0^{-1} \a \end{pmatrix}
\end{equation}
give the following decomposition:
\begin{equation}\label{eqn:LebesgueDecomp}
 \s=\s^a+\s^\perp
 \qquad (\s^a\ll \r,\;\s^\perp \perp \r)
\end{equation}
with respect to $\r$.

As in the previous subsection, we may call \eqref{eqn:LebesgueDecomp} a {\em quantum Lebesgue decomposition} for the following reasons. 
First, 
although the simultaneous block representation (\ref{eqn:simultaneousBlock}) has arbitrariness of unitary transformations of the form $U_1\oplus U_2\oplus U_3$,
the operators $\s^a$ and $\s^\perp$ are well-defined 
because the matrices (\ref{eqn:LebDec2}) are covariant under those unitary transformations.
Second, 
the decomposition (\ref{eqn:LebesgueDecomp}) is unique, as the following lemma asserts. 

\begin{lemma}\label{lem:uniqueness2}
Given $\r$ and $\s$ with $\s\not\perp\r$, 
the decomposition 
\[ \s=\s^a+\s^\perp\qquad (\s^a \ll \r,\;\s^\perp \perp \r) \]
is uniquely given by (\ref{eqn:LebDec2}).
\end{lemma}

\begin{proof}
We show that the decomposition
\begin{equation}\label{eqn:uniqueness2}
 \s=R\r R+\t\qquad (R\ge 0,\,\t\ge 0,\; \t \perp \r)
\end{equation}
is unique.
Because of Lemma \ref{lem:uniqueness1}, it suffices to treat the case when $\s\not\gg\r$, that is, when $\H_1\neq \{0\}$. 

Let $\r$ and $\s$ be represented as \eqref{eqn:simultaneousBlock}. 
It then follows from \eqref{eqn:uniqueness2} that, for any $x\in \H_1$, 
\[
 0=\bracket{x}{\s x}\ge \bracket{x}{R \r Rx}=\bracket{Rx}{\r Rx}.
\]
This implies that $Rx\in\ker\r\,(=\H_3)$: in particular, 
$\bracket{x}{Rx}=0$, so that 
the $(1,1)$th block of $R$ is zero. 
This fact, combined with the positivity of $R$, entails that $R$ must have the form 
\[
R=
\begin{pmatrix} 
0 & 0 & 0 \\ 
0 & R_0 & R_1 \\ 
0 & R_1^* & R_2  
\end{pmatrix}.
\]
Consequently, the problem is reduced to finding the decomposition
\begin{equation}\label{eqn:uniqueness2-1}
 \hat\s=\hat R \hat\r \hat R+\hat\t\qquad (\hat R\ge 0,\,\hat\t\ge 0,\; \hat\t \perp \hat\r),
\end{equation}
where
\[
\hat\r=
\begin{pmatrix} 
\r_0 & 0 \\ 
0 & 0  
\end{pmatrix},
\qquad 
\hat\s=
\begin{pmatrix} 
\s_0 & \a \\ 
\a^* & \b  
\end{pmatrix},
\qquad 
\hat R=
\begin{pmatrix} 
R_0 & R_1 \\ 
R_1^* & R_2 
\end{pmatrix}.
\]
Since $\hat{\r} \ll \hat\s$, the uniqueness of the decomposition \eqref{eqn:uniqueness2-1} has been established in Lemma \ref{lem:uniqueness1}. 
This completes the proof. 
\end{proof}

\begin{remark}
The operator $R$ appeared in the proof of Lemma \ref{lem:uniqueness2} is written as
\begin{equation}\label{eqn:remark}
R=\sqrt{\s} \left(\sqrt{\sqrt{\s} \r \sqrt{\s}} \right)^+ \sqrt{\s} +\ga,
\end{equation}
where $A^+$ denotes the generalized inverse of an operator $A$, and $\ga$ is an arbitrary positive operator that is singular with respect to $\r$. 
\end{remark}

\begin{proof}
Recall that $\s$ is decomposed as $\s=E^*\tilde\s E$, where
\[
E=
\begin{pmatrix}
I & 0 & 0\\
0 & I & \s_0^{-1} \a \\
0 & 0 & I
\end{pmatrix},
\qquad
\tilde\s=
\begin{pmatrix}
0 & 0 & 0\\
0 & \s_0 & 0 \\
0 & 0 & \beta-\alpha^* \sigma_0^{-1} \alpha
\end{pmatrix}.
\]
Then there is a unitary operator $U$ that satisfies
\[
 \sqrt{\tilde\s}\,E=U\sqrt{\s}, 
\]
and the operator $R$, modulo the singular part $R_2$, is given by
\begin{eqnarray*}
E^* \begin{pmatrix} 0 & 0 & 0 \\ 0 & \s_0 \# \r_0^{-1} & 0 \\ 0 & 0 & 0 \end{pmatrix} E
&=&
E^* \begin{pmatrix} 0 & 0 & 0 \\ 
0 & \sqrt{\s_0}\left(\sqrt{\sqrt{\s_0}\r_0\sqrt{\s_0}\,}\right)^{-1}\sqrt{\s_0} & 0 \\ 
0 & 0 & 0 \end{pmatrix} E \\
&=&
E^* \sqrt{\tilde\s}\left(\sqrt{\sqrt{\tilde\s}\r \sqrt{\tilde\s}\,}\right)^{+}\sqrt{\tilde\s}\, E\\
&=&
E^* \sqrt{\tilde\s}\left(\sqrt{\sqrt{\tilde\s} E\r E^*\sqrt{\tilde\s}\,}\right)^{+}\sqrt{\tilde\s}\, E\\
&=&
\sqrt{\s}\, U^*\left(\sqrt{ U\sqrt{\s} \r \sqrt{\s} U^*\,}\right)^{+} U \sqrt{\s} \\
&=&
\sqrt{\s}\, U^*\left(U\sqrt{\sqrt{\s} \r \sqrt{\s}\,}U^*\right)^{+} U \sqrt{\s} \\
&=&
\sqrt{\s} \left(\sqrt{ \sqrt{\s} \r \sqrt{\s}\,}\right)^{+} \sqrt{\s}.
\end{eqnarray*}
This proves the claim.
\end{proof}

\section{Quantum local asymptotic normality}\label{sec:qLAN}

In \cite{YFG}, we developed a theory of quantum local asymptotic normality (q-LAN) for models that comprise mutually absolutely continuous density operators. 
In this section we extend the scope of q-LAN to a wider class of models. 
For the reader's convenience, some basic terms and notations frequently used in q-LAN theory are summarized in Appendix.

Suppose that $\r$ is absolutely continuous with respect to $\s$, i.e., $\r\ll\s$.
Then we see from Corollary \ref{cor:uniqueness1} that the absolutely continuous part $\s^a$ 
of the quantum Lebesgue decomposition 
\[
 \s=\s^a+\s^\perp\qquad (\s^a\ll \r,\;\s^\perp \perp \r)
\]
is in fact mutually absolutely continuous, i.e., $\s_a\sim\r$.
By analogy with classical statistics \cite[Chapter 6]{Vaart}, we define the {\em quantum log-likelihood ratio} by 
\[ \L(\s |\r):=\L(\s^a |\r), \]
that is, the Hermitian operator $\L$ that satisfies 
\[
 \s^a=e^{\frac{1}{2}\L}\r e^{\frac{1}{2}\L}.
\]
This generalization enables us to get rid of the assumption of mutual absolute continuity in the theory of q-LAN. (See also the remark presented at the end of this section.)

\begin{define}[q-LAN]\label{def:QLAN}
Given a sequence $\H^{(n)}$ of finite dimensional Hilbert spaces indexed by $n\in\N$, let $\S^{(n)}=\left\{ \rho_{\theta}^{(n)} \left|\,\theta\in\Theta\subset\R^{d}\right.\right\} $
be a quantum statistical model on $\H^{(n)}$, where $\rho_{\theta}^{(n)}$
is a parametric family of density operators and $\Theta$ is an open
set. We say $\S^{(n)}$ is {\em quantum locally asymptotically normal (q-LAN)} at $\theta_{0}\in\Theta$ 
if the following conditions are fulfilled:
\begin{itemize}
\item [{\rm (i)}] $\rho_{\theta}^{(n)}\gg\rho_{\theta_{0}}^{(n)}$ for all $\theta\in\Theta$ and $n\in\N$, 
\item [{\rm (ii)}] there exist a list $\Delta^{(n)}=\left(\Delta_{1}^{(n)},\dots,\Delta_{d}^{(n)}\right)$
of observables on each $\H^{(n)}$ that satisfies
\[
\Delta^{(n)}\convd{\;\; \rho_{\theta_{0}}^{(n)}} \;N(0,J),
\]
where $J$ is a $d\times d$ Hermitian positive semidefinite matrix
with ${\rm Re}\,J>0$, 
\item [{\rm (iii)}] quantum log-likelihood ratio $\L_{h}^{(n)}
=\L\left(\left.\rho_{\theta_{0}+h/\sqrt{n}}^{(n)}\right|\rho_{\theta_{0}}^{(n)}\right)$
is expanded in $h\in\R^{d}$ as
\[
\L_{h}^{(n)}=h^{i}\Delta_{i}^{(n)}-\frac{1}{2}\left(J_{ij}h^{i}h^{j}\right)I^{(n)}+o\left(\Delta^{(n)},\rho_{\theta_{0}}^{(n)}\right),
\]
where $I^{(n)}$ is the identity operator on $\H^{(n)}$.
\end{itemize}
\end{define}

The scope of the joint quantum local asymptotic normality introduced in \cite{YFG} is also extended as follows.

\begin{define}[joint q-LAN]\label{def:jointQLAN} 
Let $\S^{(n)}=\left\{ \rho_{\theta}^{(n)}\left|\,\theta\in\Theta\subset\R^{d}\right.\right\}$
be as in Definition \ref{def:QLAN}, and let $X^{(n)}=\left(X_{1}^{(n)},\dots,X_{r}^{(n)}\right)$
be a list of observables on $\H^{(n)}$. We say the pair $\left(\S^{(n)},X^{(n)}\right)$
is {\em jointly quantum locally asymptotically normal (jointly q-LAN)} at $\theta_{0}\in\Theta$ if the following conditions are fulfilled:
\begin{itemize}
\item [{\rm (i)}] $\rho_{\theta}^{(n)}\gg\rho_{\theta_{0}}^{(n)}$ for all $\theta\in\Theta$ and $n\in\N$,  
\item [{\rm (ii)}] there exist a list $\Delta^{(n)}=\left(\Delta_{1}^{(n)},\dots,\Delta_{d}^{(n)}\right)$
of observables on each $\H^{(n)}$ that satisfies
\[
\begin{pmatrix}X^{(n)}\\
\Delta^{(n)}
\end{pmatrix}
\convd{\;\; \rho_{\theta_{0}}^{(n)}}
N\left(\begin{pmatrix}0\\
0
\end{pmatrix},\begin{pmatrix}\Sigma & \tau\\
\tau* & J
\end{pmatrix}\right),
\]
where $\Sigma$ and $J$ are Hermitian positive semidefinite matrices
of size $r\times r$ and $d\times d$, respectively, with ${\rm Re}\,J>0$,
and $\tau$ is a complex matrix of size $r\times d$. 
\item [{\rm (iii)}] quantum log-likelihood ratio $\L_{h}^{(n)}=\L\left(\left.\rho_{\theta_{0}+h/\sqrt{n}}^{(n)}\right|\rho_{\theta_{0}}^{(n)}\right)$
is expanded in $h\in\R^{d}$ as
\[
\L_{h}^{(n)}=h^{i}\Delta_{i}^{(n)}-\frac{1}{2}\left(J_{ij}h^{i}h^{j}\right)I^{(n)}+o\left(\begin{pmatrix}X^{(n)}\\
\Delta^{(n)}
\end{pmatrix},\rho_{\theta_{0}}^{(n)}\right).
\]
\end{itemize}
\end{define}

With the above definitions, we obtain the following theorem, which is regarded as a quantum extension of Le Cam's third lemma.

\begin{theorem}[quantum Le Cam third lemma]\label{thm:LeCam}
	Let $\S^{(n)}$ and $X^{(n)}$ be as in Definition \ref{def:jointQLAN}. If $\left(\S^{(n)},X^{(n)}\right)$
is jointly q-LAN at $\theta_{0}\in\Theta$, then 
\[ X^{(n)}\convd{\;\;\rho_{\theta_0+h/\sqrt{n}}^{(n)}} \;N(({\rm Re}\,\tau) h,\Sigma) \]
for $h\in\R^{d}$. 
\end{theorem}

\begin{proof}
Let $(X_{1},\,\dots,\,X_{r},\,\Delta_{1},\,\dots,\,\Delta_{d})$
be the basic canonical observables of the CCR-algebra 
$\CCR{\Im\begin{pmatrix}\Sigma & \tau\\ \tau^{*} & J \end{pmatrix}}$, 
and let $\tilde{\phi}$ be the quantum Gaussian state 
$N\left(\begin{pmatrix}0\\ 0 \end{pmatrix},\begin{pmatrix}\Sigma & \tau\\ \tau^{*} & J \end{pmatrix}\right)$ 
on that algebra. 
For a finite subset $\left\{ \xi_{t}\right\} _{t=1}^{s}$ of $\R^{r}$, we see that
\begin{eqnarray}
&&\lim_{n\rightarrow\infty}\Tr{e^{\frac{1}{2}\L_{h}^{(n)}}\rho_{\theta_{0}}^{(n)}e^{\frac{1}{2}\L_{h}^{(n)}}\left(\prod_{t=1}^{s}e^{\sqrt{-1}\xi_{t}^{i}X_{i}^{(n)}}\right)} \label{eq:bigQGaussian} \\
&&\qquad =
\lim_{n\rightarrow\infty}\Tr{\rho_{\theta_{0}}^{(n)}e^{\frac{1}{2}\L_{h}^{(n)}}\left(\prod_{t=1}^{s}e^{\sqrt{-1}\xi_{t}^{i}X_{i}^{(n)}}\right)e^{\frac{1}{2}\L_{h}^{(n)}}} \nonumber \\
&&\qquad =
\tilde{\phi}\left(e^{\frac{1}{2}\left(h^{i}\Delta_{i}-\frac{1}{2}\left(J_{ij}h^{i}h^{j}\right)\right)}\left(\prod_{t=1}^{s}e^{\sqrt{-1}\xi_{t}^{i}X_{i}}\right)e^{\frac{1}{2}\left(h^{i}\Delta_{i}-\frac{1}{2}\left(J_{ij}h^{i}h^{j}\right)\right)}\right)\nonumber \\
&&\qquad =
 e^{-\frac{1}{2}J_{ij}h^{i}h^{j}}\tilde{\phi}\left(e^{\frac{1}{2}h^{i}\Delta_{i}}\left(\prod_{t=1}^{s}e^{\sqrt{-1}\xi_{t}^{i}X_{i}}\right)e^{\frac{1}{2}h^{i}\Delta_{i}}\right)\nonumber \\
&&\qquad =
\exp\left(\sum_{t=1}^{s}\left(\sqrt{-1}\xi_{t}^{i}h^{j}\left({\rm Re}\,\tau\right)_{ij}-\frac{1}{2}\xi_{t}^{i}\xi_{t}^{j}
\Sigma_{ji}\right)-\sum_{t=1}^{s}\sum_{u=t+1}^{s}\xi_{t}^{i}\xi_{u}^{j}\Sigma_{ji}\right). \nonumber
\end{eqnarray}
The last equality is proven in a similar way to \cite[Theorem 2.9]{YFG}. Thus
\begin{eqnarray}
&&\lim_{n\rightarrow\infty}\left|\Tr{\left(\rho_{\theta_{0}+h/\sqrt{n}}^{(n)}-e^{\frac{1}{2}\L_{h}^{(n)}}\rho_{\theta_{0}}^{(n)}e^{\frac{1}{2}\L_{h}^{(n)}}\right)\left(\prod_{t=1}^{s}e^{\sqrt{-1}\xi_{t}^{i}X_{i}^{(n)}}\right)}\right| \label{eq:L1Norm} \\
&&\qquad\le 
\lim_{n\rightarrow\infty}\left|\Tr{\left(\rho_{\theta_{0}+h/\sqrt{n}}^{(n)}-e^{\frac{1}{2}\L_{h}^{(n)}}\rho_{\theta_{0}}^{(n)}e^{\frac{1}{2}\L_{h}^{(n)}}\right)}\right|  \nonumber\\
&&\qquad =
\lim_{n\rightarrow\infty}\Tr{\left(\rho_{\theta_{0}+h/\sqrt{n}}^{(n)}-e^{\frac{1}{2}\L_{h}^{(n)}}\rho_{\theta_{0}}^{(n)}e^{\frac{1}{2}\L_{h}^{(n)}}\right)}  \nonumber\\
&&\qquad =0. \nonumber
\end{eqnarray}
The last equality follows from the identity
\[
\lim_{n\rightarrow\infty}\Tr{e^{\frac{1}{2}\L_{h}^{(n)}}\rho_{\theta_{0}}^{(n)}e^{\frac{1}{2}\L_{h}^{(n)}}}=1,
\]
which is verified by setting $\xi_{t}=0$ in (\ref{eq:bigQGaussian}). 
By combining (\ref{eq:L1Norm}) with (\ref{eq:bigQGaussian}), we have 
\begin{eqnarray*}
&&\lim_{n\rightarrow\infty}\Tr{\rho_{\theta_{0}+h/\sqrt{n}}^{(n)}\left(\prod_{t=1}^{s}e^{\sqrt{-1}\xi_{t}^{i}X_{i}^{(n)}}\right)} \\
&&\qquad\qquad =
\exp\left(\sum_{t=1}^{s}\left(\sqrt{-1}\xi_{t}^{i}h^{j}\left({\rm Re}\,\tau\right)_{ij}-\frac{1}{2}\xi_{t}^{i}\xi_{t}^{j}
\Sigma_{ji}\right)-\sum_{t=1}^{s}\sum_{u=t+1}^{s}\xi_{t}^{i}\xi_{u}^{j}\Sigma_{ji}\right).
\end{eqnarray*}
Since the right-hand side is the quasi-characteristic function of $N(({\rm Re}\,\tau)h,\,\Sigma)$, we have
\[ X^{(n)}\convd{\;\;\rho_{\theta_0+h/\sqrt{n}}^{(n)}} \;N(({\rm Re}\,\tau) h,\Sigma). \]
This completes the proof.
\end{proof}

We now proceed to the i.i.d case. 
In classical statistics, it is known that the i.i.d.~extension of a model $\{P_{\theta}  \,|\, \theta \in \Theta \subset \R^d \}$ on a measure space $(\O, \F,\m)$ having densities $p_\theta$ with respect to $\mu$ is LAN at $\theta_0$ if the model is {\em differentiable in quadratic mean} at $\theta_0$ \cite[p.~93]{Vaart}, that is, if there are random variables $\ell_1,\dots,\ell_d$ that satisfy
\begin{equation}\label{eq:quadratic}
\int_\O 
\left[
\sqrt{p_{\theta_0+h}} - \sqrt{p_{\theta_0}} - \frac{1}{2} h^i \ell_i \sqrt{p_{\theta_0}}
\right]^2 d\mu = o(\|h\|^2)
\end{equation}
as $ h\rightarrow 0$.
This condition is rewritten as
\begin{equation}\label{eq:quadratic2}
\int_\O p_{\theta_0}
\left[
\sqrt{ \frac{p_{\theta_0+h}^a}{p_{\theta_0}} }-1- \frac{1}{2} h^i \ell_i 
\right]^2 d\mu 
+\int_\O p_{\theta_0+h}^\perp d\mu
= o(\|h\|^2),
\end{equation}
where
\[
 p_{\theta_0+h}^a(\o):=\left\{\array{ll} p_{\theta_0+h}(\o),& \o\in\O_0 \\ 0, & \o\notin \O_0\endarray\right.
\]
and
\[
 p_{\theta_0+h}^\perp(\o):=\left\{\array{ll} 0,& \o\in\O_0 \\ p_{\theta_0+h}(\o), & \o\notin \O_0\endarray\right.
\]
with $\O_0:=\{\o\in\O\,|\, p_{\th_0}(\o)>0\}$. 
The first term in the left-hand side of \eqref{eq:quadratic2} is connected with the differentiability of the likelihood ratio at $h=0$, while the second term with the negligibility of the singular part. 

The quantum counterpart of this characterization is given by the following

\begin{theorem}[q-LAN for i.i.d.~models]\label{thm:iid}
Let $\left\{ \rho_{\theta}\left|\,\theta\in\Theta\subset\R^{d}\right.\right\} $
be a quantum statistical model on a finite dimensional Hilbert space
$\H$ satisfying $\rho_{\theta}\gg\rho_{\theta_{0}}$ for all $\theta\in\Theta$ 
with respect to a fixed $\theta_{0}\in\Theta$. 
If $\L_{h}:=\L\left(\left.\rho_{\theta_{0}+h}\right|\rho_{\theta_{0}}\right)$
is differentiable at $h=0$, 
and the trace of the absolutely continuous part satisfies
\begin{equation}\label{eq:oh2}
\Tr{\rho_{\theta_{0}}e^{\L_{h}}}
=
1-o(\|h\|^{2}),
\end{equation}
then $\left\{ \rho_{\theta}^{\otimes n}\left|\,\theta\in\Theta\subset\R^{d}\right.\right\} $
is q-LAN at $\theta_{0}$; 
that is, 
$\rho_{\theta}^{\otimes n}\gg\rho_{\theta_{0}}^{\otimes n}$ for all $\th\in\Th$, and 
\[
\Delta_{i}^{(n)}:=\frac{1}{\sqrt{n}}\sum_{k=1}^{n}I^{\otimes(k-1)}\otimes L_{i}\otimes I^{\otimes(n-k)},
\]
satisfies (ii) and (iii) in Definition \ref{def:QLAN}. 
Here $L_{i}$ is (a version of) the $i$th symmetric logarithmic derivative at $\theta_{0}\in\Theta$, 
and $J=(J_{ij})$ is given by
\[
 J_{ij}:=\Tr{\rho_{\theta_{0}}L_{j}L_{i}}. 
\]
\end{theorem}

\begin{proof}
We first note that for positive operators $A, B, C$, and $D$ satisfying $A\ge B$ and $C\ge D$, we have 
$A\otimes C\ge B\otimes D$: in fact, 
\[
 A\otimes C-B\otimes D=A\otimes (C-D)+(A-B)\otimes D\ge 0.
\]
As a consequence, for operators $\r$ and $\s$ satisfying $\r\ge \s\ge 0$, we have
\[
 \r^{\otimes n}\ge \s^{\otimes n}
\]
for all $n\in\N$. 

Now, let $\r_\th=\r_\th^a+\r_\th^\perp$ be the quantum Lebesgue decomposition with respect to $\r_{\th_0}$. 
It then follows from the above observation that
\begin{equation}\label{eqn:absoluteContinuity-n}
\rho_{\theta}^{\otimes n} 
\ge 
\left(\rho_{\theta}^a\right)^{\otimes n} 
=
\left(
e^{\frac{1}{2}\L\left(\rho_{\theta}\left|\rho_{\theta_{0}}\!\right.\right)}
\rho_{\theta_{0}}
e^{\frac{1}{2}\L\left(\rho_{\theta}\left|\rho_{\theta_{0}}\!\right.\right)}
\right)^{\otimes n}
=
e^{\frac{1}{2}\L^{(n)}}\,\rho_{\theta_{0}}^{\otimes n}\, e^{\frac{1}{2} \L^{(n)}},
\end{equation}
where
\begin{equation}\label{eqn:qllr-n}
 \L^{(n)}
 :=\sum_{k=1}^{n}I^{\otimes(k-1)}\otimes
 \L\left(\rho_{\theta}\left|\rho_{\theta_{0}}\!\right.\right)
 \otimes I^{\otimes(n-k)}.
\end{equation}
We prove that $\L^{(n)}$ is nothing but the quantum log-likelihood ratio 
$\L\left(\rho_{\theta}^{\otimes n}\left|\rho_{\theta_{0}}^{\otimes n}\!\right.\right)$ 
between $\rho_{\theta}^{\otimes n}$ and $\rho_{\theta_0}^{\otimes n}$. 
First of all, we note that
\[
 \Tr \r_{\th_0}\r_\th= \Tr \r_{\th_0}\r_\th^a+\Tr \r_{\th_0}\r_\th^\perp=\Tr \r_{\th_0}\r_\th^a
 =\Tr \rho_{\theta_0} \left(
 e^{\frac{1}{2}\L\left(\rho_{\theta}\left|\rho_{\theta_{0}}\!\right.\right)}
 \rho_{\theta_{0}}
 e^{\frac{1}{2}\L\left(\rho_{\theta}\left|\rho_{\theta_{0}}\!\right.\right)}
 \right),
\]
which follows from Lemma \ref{lem:1}. 
By using this identity, we find that
\begin{eqnarray*}
	&&\Tr \rho_{\theta_0}^{\otimes n} \left[ \rho_{\theta}^{\otimes n} - \left(
	e^{\frac{1}{2}\L\left(\rho_{\theta}\left|\rho_{\theta_{0}}\!\right.\right)}
	\rho_{\theta_{0}}
	e^{\frac{1}{2}\L\left(\rho_{\theta}\left|\rho_{\theta_{0}}\!\right.\right)}
	\right)^{\otimes n} 
	\right]  \\
	&&\qquad=
	(\Tr \rho_{\theta_0} \rho_{\theta})^n - \left
	(\Tr \rho_{\theta_0} \left(
	e^{\frac{1}{2}\L\left(\rho_{\theta}\left|\rho_{\theta_{0}}\!\right.\right)}
	\rho_{\theta_{0}}
	e^{\frac{1}{2}\L\left(\rho_{\theta}\left|\rho_{\theta_{0}}\!\right.\right)}
	\right) \right)^n=0.
\end{eqnarray*}
In view of Lemma \ref{lem:1} again, this implies that
\begin{equation}\label{eqn:singular-n} 
	\rho_{\theta_0}^{\otimes n} \perp \left[ \rho_{\theta}^{\otimes n} - 
	\left(
	e^{\frac{1}{2}\L\left(\rho_{\theta}\left|\rho_{\theta_{0}}\!\right.\right)}
	\rho_{\theta_{0}}
	e^{\frac{1}{2}\L\left(\rho_{\theta}\left|\rho_{\theta_{0}}\!\right.\right)}
	\right)^{\otimes n} \right].
\end{equation}
From \eqref{eqn:absoluteContinuity-n} and \eqref{eqn:singular-n}, we have the quantum Lebesgue decomposition:
\[
 \r_\th^{\otimes n}=(\r_\th^{\otimes n})^a+(\r_\th^{\otimes n})^\perp
\] 
with respect to $\r_{\th_0}^{\otimes n}$, where
\[
 (\r_\th^{\otimes n})^a
 =
 e^{\frac{1}{2}\L^{(n)}}\,\rho_{\theta_{0}}^{\otimes n}\, e^{\frac{1}{2} \L^{(n)}}
\]
and
\[
 (\r_\th^{\otimes n})^\perp
 =\rho_{\theta}^{\otimes n} - 
 e^{\frac{1}{2}\L^{(n)}}\,\rho_{\theta_{0}}^{\otimes n}\, e^{\frac{1}{2} \L^{(n)}}.
\]
It is now obvious that the quantum log-likelihood ratio is given by
\begin{equation}\label{eq:iidRatio}
\L\left(\rho_{\theta}^{\otimes n} \left| \rho_{\theta_{0}}^{\otimes n}\right.\right)
=
\L\left((\rho_{\theta}^{\otimes n})^a \left| \rho_{\theta_{0}}^{\otimes n}\!\right.\right)
=\L^{(n)}.
\end{equation}

Before proceeding to the proof of (ii) and (iii) in Definition \ref{def:QLAN},
we give some preliminary consideration. 
Since $\L_{h}$ is differentiable at $h=0$, 
there are Hermitian operators $A_{i}$ ($1\leq i\leq d$) that satisfy
\begin{equation}\label{eqn:introduceSLD}
\L_{h}=h^{i}A_{i}+o(\|h\|).
\end{equation}
It is import to observe that the operator $A_i$ is (a version of) the symmetric logarithmic derivative (SLD) of the model at $\th=\th_0$ in the $i$th direction.
In fact, since the singular part 
\[
 \r_{\th_0+h}^\perp:=\r_{\th_0+h}-e^{\frac{1}{2}\L_{h}}\rho_{\th_{0}}e^{\frac{1}{2}\L_{h}}
\]
is positive, the condition (\ref{eq:oh2}) entails that
\[
\rho_{\theta_{0}+h}=e^{\frac{1}{2}\L_{h}}\rho_{\theta_{0}}e^{\frac{1}{2}\L_{h}}+o(\|h\|^2). 
\]
Substituting \eqref{eqn:introduceSLD} into this equation, we have
\begin{eqnarray*}
\rho_{\theta_{0}+h} 
&=&
\exp\left[\frac{1}{2}\left(h^{i}A_{i}+o(\|h\|)\right)\right]\rho_{\theta_{0}}\exp\left[\frac{1}{2}\left(h^{i}A_{i}+o(\|h\|)\right)\right]+o(\|h\|^{2})\\
&=&
\rho_{\theta_{0}}+\frac{1}{2}h^{i}(\rho_{\theta_{0}}A_{i}+A_{i}\rho_{\theta_{0}})+o(\|h\|),
\end{eqnarray*}
so that
\[
 \left.\frac{\partial \rho_{\theta_{0}+h}}{\partial h^i}\right|_{h=0}
 =\frac{1}{2}\left(\rho_{\theta_{0}}A_{i}+A_{i}\rho_{\theta_{0}}\right).
\]
This proves that $A_i$ is the SLD in the $i$th direction at $h=0$.
In particular, we have $\Tr{\rho_{\theta_{0}}A_{i}}=0$ for all $i$.

We next evaluate the remainder term $B(h):=\L_{h}-h^{i}A_{i}$ in \eqref{eqn:introduceSLD} in more detail. 
Observe that
\begin{eqnarray*}
\Tr{\rho_{\theta_{0}}e^{\L_{h}}}
&=&\Tr{\rho_{\theta_{0}}\exp\left(h^{i}A_{i}+B(h)\right)}\\
&=&\Tr{\rho_{\theta_{0}}\left(I+\left(h^{i}A_{i}+B(h) \right)+\frac{1}{2}\left(h^{i}A_{i}+B(h) \right)^{2}\right)}
+o(\|h\|^2)\\
&=&1+h^{i}\left(\Tr{\rho_{\theta_{0}}A_{i}}\right)+ \Tr{\rho_{\theta_{0}}B(h)} 
+\frac{1}{2}h^{i}h^{j}\Tr{\rho_{\theta_{0}}A_{i}A_{j}}+o(\|h\|^2)\\
&=&1+  \Tr{\rho_{\theta_{0}}B(h)}  +\frac{1}{2} h^{i}h^{j} J_{ji}
+o(\|h\|^2),
\end{eqnarray*}
where $J_{ji}:=\Tr{\rho_{\theta_{0}}A_{i}A_{j}}$. 
Because of  the assumption \eqref{eq:oh2},
the above equation leads to
\begin{equation}\label{eq:ABcond}
\Tr{\rho_{\theta_{0}}B(h)}  +\frac{1}{2} h^{i}h^{j}J_{ij}=o(\|h\|^2).  
\end{equation}

Now we are ready to prove (ii) and (iii).
Let 
\[
\Delta_{i}^{(n)}:=\frac{1}{\sqrt{n}}\sum_{k=1}^{n}I^{\otimes(k-1)}\otimes A_{i}\otimes I^{\otimes(n-k)}.
\]
It then follows from the quantum central limit theorem \cite{qclt} that 
\begin{equation}\label{eqn:qCLT}
 \Delta^{(n)}\convd{\;\;\rho_{\theta_0}^{\otimes n}} N(0,J).
\end{equation}
This proves (ii). 
On the other hand, we see from \eqref{eq:iidRatio} and \eqref{eqn:qllr-n} that
\[
\L_{h}^{(n)}
=\L\left.\left(\r_{\th_0+h/\sqrt{n}}^{\otimes n} \right| \r_{\th_0}^{\otimes n}\right)
=\sum_{k=1}^{n}I^{\otimes(k-1)}\otimes\L_{h/\sqrt{n}}\otimes I^{\otimes(n-k)}.
\]
Let us prove that
\[
R_{h}^{(n)}:=\L_{h}^{(n)}-\left(h^{i}\Delta_{i}^{(n)}-\frac{1}{2}\left(J_{ij}h^{i}h^{j}\right)I^{\otimes n}\right)
\]
is infinitesimal relative to the convergence \eqref{eqn:qCLT}.
It is rewritten as
\begin{eqnarray*}
	R_{h}^{(n)} & = & \sum_{k=1}^{n}I^{\otimes(k-1)}\otimes
	\left[
		\L_{h/\sqrt{n}}-\frac{1}{\sqrt{n}}h^{i}A_{i}+\frac{1}{2n}\left(J_{ij}h^{i}h^{j}\right)I
	\right]\otimes I^{\otimes(n-k)}\\
	& = & \sum_{k=1}^{n}I^{\otimes(k-1)}\otimes
	\left[
		B\left(\frac{h}{\sqrt{n}} \right)+\frac{1}{2n}\left(J_{ij}h^{i}h^{j}\right)I
	\right]\otimes I^{\otimes(n-k)}\\
	& = & \sum_{k=1}^{n}I^{\otimes (k-1)}\otimes\frac{1}{\sqrt{n}}P(n)\otimes I^{\otimes (n-k)},
\end{eqnarray*}
where
\[
P(n):=\sqrt{n}\left(
	B\left(\frac{h}{\sqrt{n}} \right)+\frac{1}{2n}\left(J_{ij}h^{i}h^{j}\right)I
\right).
\]
Note that $\lim_{n\rightarrow\infty}P(n)=0$, and that 
\begin{eqnarray*}
	\lim_{n\rightarrow\infty}\sqrt{n}\,\Tr\rho_{\theta_{0}}P(n) 
	=  \lim_{n\rightarrow\infty} \frac{\Tr \rho_{\theta_0} B(h/\sqrt{n})+\frac{1}{2n} J_{ij}h^i h^j}{(1/\sqrt{n})^2} 
	= 0
\end{eqnarray*}
because of (\ref{eq:ABcond}). 
It then follows from \cite[Lemma 2.6]{YFG}
that $R_{h}^{(n)}=o(\Delta^{(n)},\rho_{\theta_{0}}^{\otimes n})$
for any $h\in\R^{d}$. This completes the proof.
\end{proof}

The following corollary is an i.i.d.~version of the quantum Le Cam third lemma. 

\begin{corollary}[quantum Le Cam third lemma for i.i.d.~models]\label{cor:iid}
Let $\left\{ \rho_{\theta}\left|\,\theta\in\Theta\subset\R^{d}\right.\right\} $
be a quantum statistical model on a finite dimensional Hilbert space
$\H$ satisfying $\rho_{\theta}\gg\rho_{\theta_{0}}$ for all $\theta\in\Theta$
with respect to a fixed $\theta_{0}\in\Theta$. 
Further, let $\{B_i\}_{1\le i\le r}$ be observables on $\H$ satisfying $\Tr \r_{\theta_0} B_i=0$ for $i=1,\dots,r$. 
If $\L_{h}:=\L\left(\left.\rho_{\theta_{0}+h}\right|\rho_{\theta_{0}}\right)$ 
is differentiable at $h=0$, 
and the trace of the absolutely continuous part satisfies
\[
\Tr{\rho_{\theta_{0}}e^{\L_{h}}}=1-o(\|h\|^{2}),
\]
then the pair $\left(\left\{\r_\th^{\otimes n}\right\}, X^{(n)}\right)$ of i.i.d.~extension model $\left\{\r_\th^{\otimes n}\right\}$ and the list $X^{(n)}=\{X_i^{(n)}\}_{1\le i\le r}$ of observables defined by
\[
 X_i^{(n)}:=\frac{1}{\sqrt{n}}\sum_{k=1}^n I^{\otimes(k-1)}\otimes B_i\otimes I^{\otimes (n-k)}
\]
is jointly q-LAN at $\theta_{0}$, and 
\[
X^{(n)}
\convd{\rho_{\theta_{0}+h/\sqrt{n}}^{\otimes n}}
N(({\rm Re}\,\tau)h,\,\Sigma)
\]
for $h\in\R^{d}$, 
where $\Sigma$ is the $r\times r$ positive semidefinite matrix defined by $\Sigma_{ij}=\Tr\r_{\theta_0} B_jB_i$ 
and $\tau$ is the $r\times d$ matrix defined by $\tau_{ij}=\Tr\r_{\theta_0}L_j B_i$ 
with $L_{i}$ being (a version of) the $i$th symmetric logarithmic derivative at $\theta_{0}\in\Theta$.
\end{corollary}

\begin{proof}
That $\rho_{\theta}^{\otimes n}\gg\rho_{\theta_{0}}^{\otimes n}$
has been proven in Theorem \ref{thm:iid}. 
Let $\Delta_{1}^{(n)},\dots,\Delta_{d}^{(n)}$ be as in the proof of Theorem
\ref{thm:iid}. It then follows from the quantum central limit
theorem that 
\begin{equation}\label{eq:qcovergenceTogetherIID}
\begin{pmatrix}X^{(n)}\\
\Delta^{(n)}
\end{pmatrix}
\convd{\;\; \rho_{\theta_{0}}^{\otimes n}}
N\left(\begin{pmatrix}0\\
0
\end{pmatrix},\begin{pmatrix}\Sigma & \tau\\
\tau* & J
\end{pmatrix}\right).
\end{equation}
Further, because of \cite[Lemma 2.6]{YFG},
the sequence $R_{h}^{(n)}$
of observables given in the proof of Theorem \ref{thm:iid} is
also infinitesimal relative to the convergence (\ref{eq:qcovergenceTogetherIID}).
Now that $\left(\left\{\r_\th^{\otimes n}\right\}, X^{(n)}\right)$ is jointly QLAN at $\theta_{0}$,
the convergence 
\[
X^{(n)}
\convd{\rho_{\theta_{0}+h/\sqrt{n}}^{\otimes n}}
N(({\rm Re}\,\tau)h,\,\Sigma)
\]
is an immediate consequence of Theorem \ref{thm:LeCam}. 
\end{proof}

We conclude this section with some remarks. 
First, 
the technical assumption $\rho_{\theta}\gg\rho_{\theta_{0}}$
that ensures the existence of the quantum log-likelihood ratio 
$\L(\r_{\th} |\r_{\th_{0}})$
in Theorem \ref{thm:iid} or Corollary \ref{cor:iid} 
is inessential.
In fact, 
every state $\r_\th$ that is sufficiently close to $\r_{\th_0}$ satisfies $\r_\th\gg \r_{\th_0}$,
as the following lemma asserts. 

\begin{lemma}\label{lem: neighborhood}
Given a state $\r\in\S(\H)$, let $\e$ be the minimum positive eigenvalue of $\r$, 
and let $U_\e(\r):=\{\s\in\S(\H)\,|\,\|\s-\r\|<\e\}$. 
Then every state $\s\in U_\e(\r)$ satisfies $\s \gg \r$. 
\end{lemma}

\begin{proof}
For any $\s\in U_\e(\r)$, it holds that $-\e I< \s-\r <\e I$. 
Thus
\[
\s\!\!\downharpoonleft_{\supp\r}
=\r\!\!\downharpoonleft_{\supp\r}+(\s-\r)\!\!\downharpoonleft_{\supp\r}
>\r\!\!\downharpoonleft_{\supp\r}-\e I\!\!\downharpoonleft_{\supp\r}
\geq 0,
\]
proving that $\s \gg \r$.
\end{proof}

Second, 
for a quantum statistical model that fulfills the assumptions in Theorem \ref{thm:iid}, 
it is shown that the Holevo bound is asymptotically achievable at $\th_0$. 
In fact, let us take the operators $\{B_i\}_{1\le i\le r}$ in Corollary \ref{cor:iid} to be a basis of the minimal 
$\D$-invariant extension of the SLD tangent space 
at $\th_0$, where $\D$ is the commutation operator \cite{Holevo:1982}. 
Then the Holevo bound for the original model $\{\r_\th\}_\th$ at $\th=\th_0$ coincides with that for the corresponding quantum Gaussian shift model $N(({\rm Re}\tau) h,\Sigma)$ at $h=0$, and hence at any $h$.  
This fact, combined with the conclusion of Corollary \ref{cor:iid}: 
\[ 
X^{(n)}
\convd{\rho_{\theta_0+h/\sqrt{n}}^{\otimes n}}  
N(({\rm Re}\tau) h,\Sigma),
\]
enables us to construct a sequence of observables that asymptotically achieve the Holevo bound.
For a concrete construction, see the proof of \cite[Theorem 3.1]{YFG}.

\section{Example}\label{sec:example}

Let us begin by investigating the following two-dimensional spin-1/2 pure state model: 
\[
\overline{\r}_\th
:=e^{\frac{1}{2}(\th^1\s_1+\th^2\s_2-\psi(\th))}
\begin{pmatrix} 1 & 0\\ 0 & 0\end{pmatrix}
e^{\frac{1}{2}(\th^1\s_1+\th^2\s_2-\psi(\th))},
\qquad (\th=(\th^1,\th^2)\in\R^2)
\]
where 
\[
 \s_1=\begin{pmatrix} 0 & 1\\ 1 & 0\end{pmatrix},\qquad
 \s_2=\begin{pmatrix} 0 & -\sqrt{-1}\\ \sqrt{-1} & 0\end{pmatrix}
\]
are the Pauli matrices, and $\psi(\th):=\log\cosh\|\th\|$. 
This model is treated within the scope of our previous paper \cite{YFG}.  
In fact, $\overline{\r}_\th\sim \overline{\r}_0$ for all $\theta$, 
and (a version of) the quantum log-likelihood ratio $\overline{\L}_\th:=\L(\overline{\r}_\th | \overline{\r}_0)$ is given by
\[\overline{\L}_\th=\th^1\s_1+\th^2\s_2-\psi(\th). \]
Letting $h=(h^1,h^2):=(\th^1,\th^2)$, the quantum log-likelihood ratio $\overline{\L}_h$ is expanded in $h$ as
\[
\overline{\L}_h=A_ih^i+o(\|h\|),
\]
where $A_i:=\s_i$ is (a version of) the SLD of the model $\overline{\r}_\th$ at $\th=0$.
Let $X^{(n)}=(X_1^{(n)},X_2^{(n)})$ be defined by
\begin{equation}\label{eq:pureLeCamX} 
 X_i^{(n)}:=\frac{1}{\sqrt{n}}\sum_{k=1}^n I^{\otimes(k-1)}\otimes A_i\otimes I^{\otimes (n-k)}.
\end{equation}
Then it is shown that $(\{\overline{\r}_\th^{\otimes n}\}, X^{(n)})$ is jointly q-LAN at $\th=0$, and 
\begin{equation}\label{eq:pureLeCam3} 
X^{(n)}\convd{\;\;\overline{\r}_{h/\sqrt{n}}^{\otimes n}} \;
N(h,J),
\end{equation}
where
\[
J=[\Tr \overline{\r}_0 A_j A_i]_{ij}=
\begin{pmatrix}
1 & -\sqrt{-1}\\
\sqrt{-1} & 1
\end{pmatrix}.
\]
For details, see \cite[Example 3.3]{YFG}.  

Now, let us consider a perturbed model: 
\[
 \r_\th:=e^{-f(\th)} \overline{\r}_\th+(1-e^{-f(\th)})\overline{\r}_*, \qquad (\th\in\R^2)
\]
where 
\[
\overline{\r}_*=\begin{pmatrix} 0 & 0\\ 0 & 1\end{pmatrix},
\]
and $f(\th)$ is a smooth function that is positive for all $\th\neq 0$ and exhibits 
\[ f(\th)=o(\|\th\|^2). \] 
Geometrically, this model is tangential to the Bloch sphere at the north pole $\r_0\,(=\overline{\r}_0)$, 
thus having a singularity at $\th=0$ in that the rank of the model drops there. 
It was customary to avoid such a singular model in the conventional quantum state estimation theory.
We shall demonstrate that this model can be treated within the framework of the present paper. 

Since
\[
 \r_\th
 \ge e^{-f(\th)}\overline{\r}_\th
 =e^{\frac{1}{2}(\overline{\L}_\th-f(\th)I)}\, \r_0\, e^{\frac{1}{2}(\overline{\L}_\th-f(\th)I)},
\]
we see from Lemma \ref{lem:2} that $\r_\th\gg \r_0$ for all $\th$. 
It is also easily seen that the quantum Lebesgue decomposition 
\[
 \r_\th=\r_\th^a+\r_\th^\perp
\]
with respect to $\r_0$ is given by
\[
 \r_\th^a:=e^{-f(\th)}\overline{\r}_\th,\qquad \r_\th^\perp:=(1-e^{-f(\th)})\overline{\r}_*. 
\]
The (mutually) absolutely continuous part $\r_\th^a$ gives (a version of) the quantum log-likelihood ratio $\L_\th:=\L(\r_\th |\r_0):=\L(\r_\th^a |\r_0)$ as
\[
 \L_\th=\overline{\L}_\th-f(\th)I.
\]
Since $f(h)=o(\|h\|^2)$, we have
\[
 \L_h=\overline{\L}_h-f(h)I=A_ih^i+o(\|h\|),
\]
where $A_i:=\s_i$ is again (a version of) the SLD of the model $\r_\th$ at $\th=0$.
On the other hand, the singular part $\r_\th^\perp$ tells us that
\begin{eqnarray*}
\Tr \r_h^\perp = o(\| h \|^2). 
\end{eqnarray*}
This ensures the condition \eqref{eq:oh2} for the model $\r_\th$ to be q-LAN at $\th=0$ (Theorem \ref{thm:iid}). 
It then follows from Corollary \ref{cor:iid} that the sequence $X^{(n)}$ of observables defined by \eqref{eq:pureLeCamX} exhibits
\begin{equation}\label{eq:pureLeCam3_2}
X^{(n)}\convd{\;\;\r_{h/\sqrt{n}}^{\otimes n}} \;
N(h,J).
\end{equation}

In summary, as far as the observables $X^{(n)}=(X^{(n)}_1,X^{(n)}_2)$ are concerned,
the i.i.d.~extension 
$\left\{\r^{\otimes n}_{h/\sqrt{n}} \left|\, h\in\R^2\right.\right\}$
of the perturbed model $\r_\th$ around the singular point $\th=0$ 
is asymptotically similar to the quantum Gaussian shift model $\{ N(h,J) \left|\, h\in\R^2 \right.\}$ 
as shown in \eqref{eq:pureLeCam3_2}, 
and is also asymptotically similar to the i.i.d.~extension 
$\left\{\overline{\r}^{\otimes n}_{h/\sqrt{n}} \left|\, h\in\R^2\right.\right\}$
of the unperturbed pure state model $\overline{\r}_\th$ around $\th=0$ as shown in \eqref{eq:pureLeCam3}.

\section{Concluding remarks}\label{sec:conclusion}

We have developed a theory of local asymptotic normality in the quantum domain based on a noncommutative extension of the Lebesgue decomposition. 
This formulation is applicable to models that do not necessarily comprise mutually absolutely continuous density operators, thus allowing singularity at the reference state. 
In this respect, the present paper gives a substantial generalization of  our previous paper \cite{YFG}. 

However, there are still many open problems left. 
Among others, it is not clear whether every sequence of positive operator-valued measures on a q-LAN model can be realized on the limiting quantum Gaussian shift model. 
In classical statistics, this question has been solved affirmatively as the representation theorem \cite{Vaart}, 
which asserts that, given a weakly convergent sequence $T^{(n)}$ of statistics on 
$\left\{p^{(n)}_{\th_0+h/\sqrt{n}} \left|\, h\in\R^d \right.\right\}$,
there exist a limiting statistics $T$ on the Gaussian shift model $\left\{ N(h,J^{-1}) \left|\, h\in\R^d \right.\right\}$
such that $T^{(n)} \convd{h} T$. 
Representation theorem is particularly useful in proving
the non-existence of an estimator that can asymptotically do better than what can be achieved in the limiting Gaussian shift model. 
Extending the representation theorem
to the quantum domain is an important open problem. 

\section*{Acknowledgment}
 
The present study was supported by JSPS KAKENHI Grant Number JP22340019.

\appendix
\section*{Appendix: Terms and notations}

Given a $d\times d$ real skew-symmetric matrix $S=[S_{ij}]$, let $\CCR{S}$ denote the algebra generated by the observables $X=(X_1,\dots,X_d)$ that satisfy the following canonical commutation relations (CCR): 
\[ \frac{\sqrt{-1}}{2}[X_i,X_j] = S_{ij} \qquad(1\leq i,j\leq d). \]
A state $\phi$ on CCR($S$) is called a {\em quantum Gaussian state}, denoted by $\phi\sim N(h,J)$, if the characteristic function 
${\cal F}_{\xi}\{\phi\}:=\phi(e^{\sqrt{-1}\xi^{i}X_{i}})$ takes the form
\[ {\cal F}_{\xi}\{\phi\}=e^{\sqrt{-1}\xi^{i}h_{i}-\frac{1}{2}\xi^{i}\xi^{j}V_{ij}} \]
where $\x=(\x^i)_{i=1}^d\in\R^d$, $h=(h_i)_{i=1}^d\in\R^d$, and $V=[V_{ij}]$ is a real symmetric matrix such that the Hermitian matrix $J:=V+\sqrt{-1}S$ is positive semidefinite. 
When the canonical observables $X$ need to be specified, we also use the notation $(X,\phi)\sim N(h,J)$. 

When we discuss relationships between a quantum Gaussian state $\phi$ on a CCR and a state on another algebra, we need to use the {\em quasi-characteristic function}
\begin{eqnarray*}
\phi\left(\prod_{t=1}^{r} e^{\sqrt{-1}\xi_{t}^{i}X_{i}}\right)
=
\exp\left(\sum_{t=1}^{r}\left(\sqrt{-1}\xi_{t}^{i}h_{i}-\frac{1}{2}\xi_{t}^{i}\xi_{t}^{j}J_{ji}\right)-\sum_{t=1}^{r}\sum_{u=t+1}^{r}\xi_{t}^{i}\xi_{u}^{j}J_{ji}\right)
\end{eqnarray*}
of a quantum Gaussian state, where $(X,\phi)\sim N(h,J)$ and $\{\x_t\}_{t=1}^r$ is a finite subset of $\C^d$ \cite{qclt}. 

Given a sequence $\H^{(n)}$ of finite dimensional Hilbert spaces indexed by $n\in\N$, let
$X^{(n)}=(X_{1}^{(n)},\,\dots,\,X_{d}^{(n)})$ and $\rho^{(n)}$ be a list of observables and a density operator on each $\H^{(n)}$.
We say the sequence $\left(X^{(n)},\rho^{(n)}\right)$ {\em converges in law to a quantum Gaussian state} $N(h,J)$, in symbols
\[ 
(X^{(n)},\rho^{(n)})\convq qN(h,J)
\qquad \mbox{or simply}\qquad 
X^{(n)}\convd{\;\;\rho^{(n)}} \;N(h,J),
\]
if
\[
\lim_{n\rightarrow\infty}\Tr\rho^{(n)}\left(\prod_{t=1}^{r} e^{\sqrt{-1}\xi_{t}^{i}X_{i}^{(n)}}\right)=\phi\left(\prod_{t=1}^{r} e^{\sqrt{-1}\xi_{t}^{i}X_{i}}\right)
\]
for any finite subset $\{\xi_{t}\}_{t=1}^{r}$ of $\C^{d}$, where $(X,\phi)\sim N(h,J)$. 

Given a list $X^{(n)}=(X_{1}^{(n)},\,\dots,\,X_{d}^{(n)})$ of observables and a state $\rho^{(n)}$ on each $\H^{(n)}$ that satisfy 
$X^{(n)}\convd{\;\rho^{(n)}} N(0,J)\sim\left(X,\phi\right)$, 
we say a sequence $R^{(n)}$ of observables, each being defined on $\H^{(n)}$, is \textit{infinitesimal relative to the convergence} $X^{(n)}\convd{\;\rho^{(n)}} N(0,J)$ if it satisfies 
\begin{equation}
\lim_{n\rightarrow\infty}\Tr\rho^{(n)}\left(\prod_{t=1}^{r} e^{\sqrt{-1}\left(\xi_{t}^{i}X_{i}^{(n)}+\eta_{t}R^{(n)}\right)}\right)=\phi\left(\prod_{t=1}^{r} e^{\sqrt{-1}\xi_{t}^{i}X_{i}}\right)\label{eq:infinitesimal}
\end{equation}
for any finite subset of $\left\{ \xi_{t}\right\} _{t=1}^{r}$ of $\C^{d}$ and any finite subset $\left\{ \eta_{t}\right\} _{t=1}^{r}$ of $\C$.
An infinitesimal object $R^{(n)}$ relative to $X^{(n)}\convd{\;\rho^{(n)}} N(0,J)$ is denoted as $o(X^{(n)},\rho^{(n)})$.



\end{document}